\documentclass[12pt]{article}
\usepackage{epsfig}
\usepackage{graphicx}
\usepackage{color}
\usepackage{amsfonts,amsmath}
\usepackage{array}
\RequirePackage{pgf}
\RequirePackage{tikz}
\RequirePackage{xkeyval}

\tikzstyle dtree=[grow'=up,sibling distance=2mm,level distance=2mm,thick]
\tikzstyle dtree node=[scale=0.3,shape=circle,very thin,draw]
\tikzstyle dtree black node=[style=dtree node,fill=black]
\tikzstyle dtree white node=[style=dtree node,fill=white]

\tikzstyle ftree=[grow'=right,sibling distance=3mm,level distance=3mm,thick]
\tikzstyle ftree node=[scale=0.4,shape=circle,very thin,draw]
\tikzstyle ftree black node=[style=ftree node,fill=black]
\tikzstyle ftree white node=[style=ftree node,fill=white]

\renewcommand{\o}{
  \begin{tikzpicture}[ftree]
    \node[ftree black node] {} ;
  \end{tikzpicture}
 }

\newcommand{\oo}{
  \begin{tikzpicture}[ftree]
    \node[ftree black node] {} 
    child { node[ftree black node]{ }   
     }
    ;
  \end{tikzpicture}
 }

\renewcommand{\oe}{
  \begin{tikzpicture}[ftree]
    \node[ftree black node] {} 
    child {}
    ;
  \end{tikzpicture}
 }

\newcommand{\ooo}{
  \begin{tikzpicture}[ftree]
    \node[ftree black node] {} 
    child { node[ftree black node]{ }   
             child { node[ftree black node]{ }   
                      }
              };
  \end{tikzpicture}
 }

\newcommand{\ooe}{
  \begin{tikzpicture}[ftree]
    \node[ftree black node] {} 
        child { node[ftree black node]{ }   
                 child {}
                };
  \end{tikzpicture}
 }

\newcommand{\eoe}{
  \begin{tikzpicture}[ftree]
    \node[ftree black node] {} 
    child[grow=left] { }  
    child[grow=right] { }   
    ;
  \end{tikzpicture}
 }

\newcommand{\eooe}{
  \begin{tikzpicture}[ftree]
    \node[ftree black node] {} 
     child[grow=right]{ node[ftree black node]{ } 
               child[grow=right] { }  
               }
    child[grow=left] { }   
    ;
  \end{tikzpicture}
 }

\newcommand{\oooe}{
  \begin{tikzpicture}[ftree]
    \node[ftree black node] {} 
      child { node[ftree black node]{ }   
               child { node[ftree black node]{ }   
                      child {}
                        }
                };
  \end{tikzpicture}
 }

\newcommand{\oope}{
 \begin{tikzpicture}[ftree]
    \node[ftree black node] {} 
    child[grow=left] { node[ftree black node]{} }
    child[grow=down] {  }
    child[grow=right] { node[ftree black node]{} }
    ;
  \end{tikzpicture}
 }

\newcommand{\eope}{
 \begin{tikzpicture}[ftree]
    \node[ftree black node] {} 
    child[grow=left] {  }
    child[grow=up] { node[ftree black node]{} }
    child[grow=right] {  }
    ;
  \end{tikzpicture}
 }

\newcommand{\eepe}{
 \begin{tikzpicture}[ftree]
    \node[ftree black node] {} 
    child[grow=left] {  }
    child[grow=up] {  }
    child[grow=right] {  }
    ;
  \end{tikzpicture}
 }

\newcommand{\eoepe}{
 \begin{tikzpicture}[ftree]
    \node[ftree black node] {} 
    child[grow=left]{  node[ftree black node]{}
     child[grow=left] { }
      }
    child[grow=up] {  }
    child[grow=right] {  }
    ;
  \end{tikzpicture}
 }

\newcommand{\epope}{
 \begin{tikzpicture}[ftree]
    \node[ftree black node] {} 
  child[grow=left] {  }
    child[grow=up] { node[ftree black node]{} }
    child[grow=down] {  }
    child[grow=right] {  }
    ;
  \end{tikzpicture}
 }
 
\newcommand{\epepe}{
 \begin{tikzpicture}[ftree]
    \node[ftree black node] {} 
    child[grow=left] {  }
    child[grow=up] {   }
    child[grow=down] {  }
    child[grow=right] {  }
    ;
  \end{tikzpicture}
 }

\newcommand{\epoepe}{
 \begin{tikzpicture}[ftree]
    \node[ftree black node] {}
     child[grow=right]{}
     child[grow=72]{}    
    child[grow=144]{ }
    child[grow=-72]{}    
    child[grow=-144]{ };
      \end{tikzpicture}
 }

\newcommand{\rta}{
\begin{tikzpicture}[dtree]
    \node[dtree black node] {}
    ;
  \end{tikzpicture}
}

\newlength{\treewidth}
\newlength{\treeheight}

\def\functionspace{\mathfrak}

\def\span{\mathop{\rm span}}

\def\P{\mathfrak P}

\def\H{\mathcal H}
\def\T{\mathcal T}

\def\FT{\mathcal FT}
\def\B{\mathcal B}
\def\W{\mathcal W}
\def\X{\mathcal X}
\def\Y{\mathcal Y}

\def\R{{\mathbb R}}

\def\Cinf{C^{\infty}(\R^d)}
\def\Cinfbi{C^{\infty}(\R^{2d})}
\def\Cinfpol{C^{\infty}_{\mathrm{pol}}(\R^{2d})}

\def\End{\mathrm{End}(\Cinf)}

\def\X{{\functionspace X}}

\def\dim{\mathop{\rm dim}}

\newtheorem{definition}{Definition}
\newtheorem{conjecture}{Conjecture}
\newtheorem{proposition}{Proposition}
\newtheorem{lemma}{Lemma}

\newtheorem{theorem}{Theorem}
\newtheorem{remark}{Remark}

\newcommand{\QED}{\hspace*{2em}\hfill$\square$}
\newcommand{\order}{\mathop{\rm order}}

\newcommand{\degree}{\mathop{\rm degree}}
\newenvironment{proof}[1]{\vspace{-2pt}{\em Proof #1.\/ }}{\QED
  \vspace{8pt}}

\newcommand{\gray}[1]{\textcolor{gray}{#1}}

\def\eqalign#1{\null\,\vcenter{\openup\jot \mathsurround=0pt
  \ialign{\strut\hfil$\displaystyle{##}$&$
        \displaystyle{{}##}$\hfil \crcr#1\crcr}}\,}

\title{The Lie algebra of classical mechanics}
\author{Robert I. McLachlan\thanks{School of 
Fundamental Sciences, 
Massey University, Palmerston North, New Zealand 
({\tt R.McLachlan@massey.ac.nz}).}
\and Ander Murua
}
\begin{document}

\maketitle

\begin{abstract}
\noindent
Classical mechanical systems are defined by their kinetic and potential energies.
They generate a Lie algebra under the canonical Poisson bracket. This Lie algebra,
which is usually infinite dimensional, is useful in analyzing the system, as well
as in geometric numerical integration. But because the kinetic energy is 
quadratic in the momenta, the Lie algebra obeys identities beyond those
implied by skew symmetry and the Jacobi identity. Some Poisson brackets,
or combinations of brackets, are zero for all choices of kinetic and
potential energy, regardless of the dimension of the system. 
Therefore, we study the universal object in this setting, the
`Lie algebra of classical mechanics' modelled on the
Lie algebra generated by kinetic and potential energy of a simple
mechanical system  with respect to the canonical Poisson bracket.
We show that
it is the direct sum of an abelian algebra $\X$, spanned by `modified'
potential energies isomorphic to the free commutative nonassociative algebra with one generator, 
and an algebra freely generated by the kinetic energy and its Poisson bracket 
with $\X$.
We calculate the dimensions $c_n$ of its homogeneous
subspaces
and determine the value of its entropy $\lim_{n\to\infty} c_n^{1/n}$.
It is  $1.8249\dots$, a fundamental constant associated to
classical mechanics. We conjecture that the class of systems with 
Euclidean kinetic energy metrics is already free, i.e., the 
only linear identities satisfied by the Lie brackets of all such systems 
are those satisfied by the Lie algebra of classical mechanics.
\end{abstract}

%\begin{keywords} 
%
%\end{keywords}

%\begin{AMS}
%17B01,70G45 % free Lie algebras, & geom structures of class mech
%\end{AMS}
%\pagestyle{myheadings} \thispagestyle{plain}
%\markboth{LEFT}{RIGHT}

\def\eqalign#1{\null\,\vcenter{\openup\jot \mathsurround=0pt
  \ialign{\strut\hfil$\displaystyle{##}$&$
        \displaystyle{{}##}$\hfil \crcr#1\crcr}}\,}

\def\epsilon{\varepsilon}
\def\phi{\varphi} 
\def\A{{\cal{A}}}
\def\fix{\mathop{\rm fix}}

\def\smallmatrix{\null\,\vcenter\bgroup \baselineskip=7pt
 \ialign\bgroup\hfil$\scriptstyle{##}$\hfil&&$\;$\hfil
 $\scriptstyle{##}$\hfil\crcr}
\def\endsmallmatrix{\crcr\egroup\egroup\,}
\def\smatrix#1{\smallmatrix #1 \endsmallmatrix}
\def\spmatrix#1{\left( \smallmatrix #1 \endsmallmatrix \right)}

\section{Introduction}

Simple mechanical systems are defined by pairs $(Q,V)$, where
the configuration space $Q$ is a real Riemannian manifold and the 
potential energy $V$ is a smooth real function on $Q$.
The phase space $T^*Q$ has a canonical
Poisson bracket and a kinetic energy $T:T^*Q\to \R$ associated with the
metric on $Q$. Given two distinguished functions,
namely the kinetic and potential energies, one can ask what 
Lie algebra they generate under the Poisson bracket. 

In this paper we study, not the Lie algebra generated by a {\it particular}
$V$ and $T$, but the Lie algebra defined by the whole {\it class}
of simple mechanical systems. That is, one should think of the dimension
of $Q$ as being arbitrarily large, and the metric and potential energies
also being arbitrary. This `Lie algebra of classical mechanics' was
introduced in \cite{mc-ry}. However, that paper has an error that is corrected here.

Numerical integrators based on splitting and
composition are widely used in applications including molecular, celestial, 
and accelerator dynamics \cite{ha-lu-wa,mc-qu}. 
The vector field $X$ which is to be integrated
is split as $X=A+B$, where $A$ and $B$ have the same properties (e.g.
Hamiltonian) as $X$, but can be integrated exactly. The integrator is a composition of the form
\begin{equation}
\label{eq:higher}
 \prod_{i=1}^s \exp(a_i \tau A)\exp(b_i \tau B) = \exp(Z)
\end{equation}
where $\Delta t$ is the time step and $\exp(tX)$ is the time-$t$ flow of $X$.
The Baker--Campbell--Hausdorff formula gives
$Z\in L(A,B)$, the free Lie algebra with two generators.
Requiring $Z=\Delta t(A+B)+O((\Delta t)^{p+1})$ for
some integer $p>1$ gives a system of equations in the $a_i$ and $b_i$
which must be satisfied for the method to have order $p$. In the
case of general $A$ and $B$, then, at each order $n=1,\dots,p$ there are
$\dim L^n(A,B)$ such {\it order conditions}. Here $L^n(A,B)$ is the
subspace of $L(A,B)$ consisting of homogeneous elements of order $n$.

In this approach it is assumed that the only identities 
satisfied by Lie brackets of 
$A$ and $B$ are those due to antisymmetric and the Jacobi identity.
This is reasonable if one wants the method to work
for all $A$ and $B$. However, in the case of simple mechanical systems,
the Lie algebra is {\it never} free,
regardless of $T$, $V$, or the dimension
of the system. There are always identities satisfied by the Poisson brackets of the kinetic and potential
energy. The simplest
of these is 
\begin{equation}
\label{eq:vvvt}
\{V,\{V,\{V,T\}\}\}\equiv 0.
\end{equation} 
For, working in local
coordinates $(p,q)$ with $T=\frac{1}{2}p^T M(q) p$, and recalling
the canonical Poisson bracket
$\{A,B\} := \sum_i
\frac{\partial A}{\partial q_i}\frac{\partial B}{\partial p_i} -
\frac{\partial A}{\partial p_i}\frac{\partial B}{\partial q_i},$
we have that
$$ \{V,T\} = \sum_{i,j} \frac{\partial V}{\partial q_i} M_{ij}(q) p_j $$
is of degree 1 in $p$, and that
\begin{equation}
\label{eq:vvt}
\{V,\{V,T\}\} = \sum_{i,j}\frac{\partial V}{\partial q_i} M_{ij} 
\frac{\partial V}{\partial q_j}
\end{equation}
is a function of $q$ only. So $V$ and $\{V,\{V,T\}\}$ commute.

Thus, it was realized early on \cite{mclachlan} that in deriving 
high-order integrators as in Eq. (\ref{eq:higher}) for simple
mechanical systems, the order conditions
corresponding to $\{V,\{V,\{V,T\}\}\}$ and to all its higher Lie brackets can
be dropped. This means that more efficient integrators can be designed
for this class of systems. Much work has been done on this special case,
both because of its intrinsic theoretical and practical 
importance, and because it allows such big improvements over the 
general case. For example,
one can design special (`corrector' or `processor') 
methods of the form $\phi\psi\phi^{-1}$
\cite{bl-ca-ro2},
special methods for nearly-integrable systems such as the solar system
\cite{bl-ca-ro2,wi-ho-to},
special methods involving exact evaluation of the forces associated
with the `modified potential' (Eq. \ref{eq:vvt}) \cite{bl-ca-ro3}, and so on---see
\cite{mc-qu} for a survey. All of these studies
rely on the structure of the Lie algebra generated by kinetic and potential
energy.
Bases for the homogeneous subspaces of this Lie algebra have been constructed
for small orders \cite{bl-ca-ro3,bl-mo,murua}, and McLachlan
and Ryland \cite{mc-ry} attempted to construct the entire Lie algebra.
However, their construction is in error starting at order 13.

The paper is structured as follows. In Section \ref{ss:pgla}, we define our object
of study, the Lie algebra of classical mechanics. It is the universal Lie algebra
with two generators equipped with a grading by degree. (In the motivating example,
 the homogeneous components are homogeneous polynomials in $p$, and the two 
 generators have degree 2 (the kinetic energy) and degree 0 (the potential energy).)
 We explicitly construct this Lie algebra as the direct sum of an abelian and a free Lie algebra
 (\S\ref{ss:char})
 and specify an explicit generalized Hall basis (\S\ref{ss:basis}). In \S\ref{ss:real}
 we show that the motivating example has no further identities in general.
 
 Sections \ref{s:euclidean} and \ref{s:schrodinger} consider two special cases. The first is that the kinetic energy
 is equal to $\frac{1}{2}p\cdot p$ but the potential energy is arbitrary. In this case we conjecture,
 but are unable to prove, that there are no further identities. We provide some
 supporting evidence for the conjecture. The second, related, case is associated to the linear
 Schr\"odinger equation. We show that the universal Lie algebra generated by commutators of
 a potential and a Euclidean 
 Laplacian,  graded by degree of differential
 operators, is identical to the Euclidean ($p\cdot p$) case.
 
 Finally, in Section \ref{s:entropy} we describe the Lie algebra of classical mechanics
 quantitatively, enumerating the dimensions of its homogeneous components
 and their asymptotic rate of growth. Recall that in the free case, $\dim L^n(A,B)
 \sim \frac{1}{n}2^n$, so that the `entropy' of $L(A,B)$ is equal to 2.
 We calculate that the entropy of the Lie algebra of classical mechanics  is $1.8249111600523655937\ldots$, 
 a fundamental constant associated to classical mechanics.

\section{The free Lie algebra $L_\P(A,B)$ in
the class $\P$}
\label{s:lacm}
\subsection{Polynomially graded Lie algebras and the Lie algebra of classical mechanics}
\label{ss:pgla}

For simple mechanical systems, 
every Lie bracket of $T$ and $V$ is a homogeneous polynomial in $p$.
Furthermore,
the degrees of these polynomials 
combine in a natural way. We therefore introduce the following class $\P$
of Lie algebras.

We use the notation $[XY] := [X,Y]$, $[XYZ]:=[X,[Y,Z]]$, $[X^n Y] = [X,[X^{n-1},Y]]$, and for
sets $\X$, $\Y$, $[\X\Y] :=[\X,\Y] := \{[X,Y]\colon X\in\X,Y\in\Y\}$.

\begin{definition}
A Lie algebra $L$ is of class $\P$ (`polynomially graded') if it is graded, i.e.
$L=\bigoplus_{n\ge 0}L_n$, and its homogeneous subspaces $L_n$ satisfy
\begin{equation}
\label{eq:degree}
\eqalign{&[L_n,L_m]\subseteq L_{n+m-1}\hbox{\rm\ 
if $n>0$ or $m>0$; and} \cr
   &[L_0,L_0] = 0\cr}
\end{equation}
Note that this implies
\begin{equation}
\label{eq:zero}
[(L_0)^{n+1}L_n]=0
\end{equation}
for all $n$. 
We call the grading of $L$ its {\rm grading by degree}.
\end{definition}

We also need the concept of a Lie algebra which is free in a certain class.

\begin{definition} Let $F$ be a Lie algebra of class $\P$
generated by a set $\X$. Then $F$ is called a free Lie algebra in
the class $\P$ over the set $\X$, if for any Lie algebra
$R$ of class $\P$, every mapping $\X\to R$ that respects the degrees of the elements can be extended to 
a unique homomorphism $F \to R$ of Lie algebras
of class $\P$.
We denote as  $L_\P(A,B)$ the free Lie algebra in the class $\P$
where $A$ has degree 2 and $B$ has degree 0, and we call it the
{\rm Lie algebra of classical mechanics}.
\end{definition}

Consider the grading by degree of the free Lie algebra $L(A,B)$  as follows:
\begin{definition}
\label{def:gradingbydegree}
$L(A,B) = \bigoplus_{n \geq 0} L^n(A,B)$ such that $A \in L^2(A,B)$, $B \in L^0(A,B)$, and
\begin{equation*}
[L^n(A,B),L^m(A,B)]\subseteq L^{\min(n+m-1,0)}(A,B), \quad \mbox{for all} \quad n, m \geq 0.
\end{equation*}
We write $\degree(Y)=n$ if $Y \in L^{n}(A,B)$. 
\end{definition}

Note that we have
 \begin{equation*}
 L_{\P}(A,B)= L(A,B)/\mathcal{I},
 \end{equation*}
 where $\mathcal{I}$ is the Lie ideal generated by $[L^0(A,B),L^0(A,B)]$,
 and that  $L_{\P}(A,B)$ inherits a grading by degree, i.e. $L_{\P}(A,B) = \bigoplus_{n\geq 0} L^n_{\P}(A,B)$.

In addition to the grading by degree, $L_\P(A,B)$ also carries
the standard grading which we call the grading by {\it order},
generated by $\order(A)=\order(B)=1$ and 
$\order([Y,Z]) = \order(Y)+\order(Z)$.

Given smooth functions $T(p,q) = \frac{1}{2}M(q)(p,p)$ and $V = V(q)$ ($p,q \in \R^d$), let us denote by $\Cinfpol$ the Lie algebra under the Poisson bracket of smooth  functions $\Cinfbi$ that depend polynomially on $p$. Then  $\Cinfpol$ is of class $\P$, with $T$ of degree 2 and $V$ of degree 0.  
\begin{definition}
\label{def:Phi}
 Given smooth functions $T(p,q) = \frac{1}{2}M(q)(p,p)$ and $V = V(q)$ ($p,q \in \R^d$),  
\begin{equation*}
%\label{eq:Phi}
\Phi_{T,V}:L_\P(A,B) \to \Cinfpol
\end{equation*}
is the unique homomorphism of Lie algebras of class $\P$ such that $\Phi_{T,V}(A)=T$ and $\Phi_{T,V}(B)=V$. 
\end{definition}

\subsection{Characterization of the Lie algebra of classical mechanics}
\label{ss:char}

Given a set $\Y$, $L(\Y)$ denotes the free Lie algebra over the set $\Y$. Given two disjoint sets $\Y_1$ and $\Y_2$, the Lazard factorization of
free Lie algebras \cite{lazard} states that
$$L(\Y_1\cup\Y_2) \cong L(\Y_2) \oplus L(\cup_{n\ge 0}[\Y_2^n \Y_1]).$$
The generalized Hall basis~\cite{Reutenauer} for $L(\Y)$ (constructed in \S\ref{ss:basis}) can be obtained by iteratively eliminating elements from the generating set.  

We construct $L_\P(A,B)$ by successively applying a Lazard factorization to eliminate the elements of degree 0. We begin by first eliminating $B$,  which introduces new elements of degree 0 into the generating set which are eliminated at the next stage of an iteration. The minimum order of these new elements of degree 0 increases with each step of the iteration, allowing a passage to the limit. This leads to the following characterization of the Lie algebra of classical mechanics. We postpone its proof to Subsection~\ref{ss:basis}.

\begin{theorem}
\label{th:1}
The Lie algebra $L_\P(A,B)$ of classical mechanics satisfies
\begin{equation}
\label{eq:lpab}
L_\P(A,B) \cong \span(\X) \oplus L(A,[\X,A]),
\end{equation}
where 
$\span(\X)$
is abelian, of degree 0, and is isomorphic as a vector space to 
the free commutative nonassociative algebra generated by $B$
under the operation 
\begin{equation}
\label{eq:star}
U_1*U_2 := [U_2,[U_1,A]].
\end{equation}
\end{theorem}

Specifically,
$$\X=\{B,\ B*B,\ (B*B)*B,\ (B*B)*(B*B), ((B*B)*B)*B,\dots\},$$
with one basis element associated to each rooted full binary tree (i.e., each node has degree 0 or 2). Because of the $A$ in (\ref{eq:star}), the order of an element with $n$ $B$s is $2n-1$. 

Note that the Jacobi identity
$$ [U_2,[U_1,A]] + [U_1,[A,U_2]] + [A,[U_2,U_1]]= 0$$
reduces in the case $\degree(U_1) = \degree(U_2) = 0$ to
$$ [U_2,[U_1,A]] - [U_1,[U_2,A]] = 0,$$
that is, $U_2*U_1 = U_1*U_2$.

In the Lie algebra generated by specific  $T(p,q) = \frac{1}{2}M(q)(p,p)$ and $V=V(q)$, that is, in the image by the homomorphism $\Phi_{T,V}$ given in Definition~\ref{def:Phi},
 for each pair $(U_1,U_2) \in \X \times \X$, if $V_1 = \Phi_{T,V}(U_1)$ and $V_2 = \Phi_{T,V}(U_2)$,  we have  in coordinates, 
$$\Phi_{T,V}(U_1 * U_2) = \{V_2,\{V_1,T\}\} = M(q)(\nabla V_1(q),\nabla V_2(q)).$$
In the Euclidean case $T = \frac{1}{2}p\cdot p$, we have 
$$\Phi_{T,V}(U_1 * U_2) = \{V_2,\{V_1,T\}\} =\nabla V_1(q) \cdot \nabla V_2(q).$$

\subsection{Basis for the Lie algebra of classical mechanics}
\label{ss:basis}

Given a totally ordered set $(\Y, >)$,  generalized Hall basis are totally ordered basis $(\H, >)$ for the free Lie algebra $L(\Y)$ satisfying the following conditions~\cite{Reutenauer}:
\begin{enumerate}
\item  $\Y$ is a totally ordered subset of the totally ordered set $\H$,
\item $\H = \bigcup_{n\geq 1} \H_n$ where 
\begin{equation*}
\H_1 = \Y, \quad \H_2 =  \{[Y_2,Y_1]\ : \ Y_2 < Y_1, \ Y_1, Y_2 \in \Y \},
\end{equation*}
 and for $n \geq 3$, $U \in \H_n$ if and only if $U =  [U_3 U_2 U_1]$ where 
\begin{equation*}
% U_1, U_2, 
  U_3 ,  [U_2, U_1]  \in \H, \  [U_2,U_1] > U_3 \geq U_2, \  \order(U_1)+\order(U_2) + \order(U_3)=n.
\end{equation*}
%\begin{equation*}
%\H_n = \{  [U_3 U_2 U_1]\ : \  [U_2,U_1] > U_3 \geq U_2,  \ U_1, U_2,  U_3 ,  [U_2, U_1]  \in \H,\ |U_1| + |U_2| + |U_3| = n \},
%\end{equation*}
\item If $U_1, U_2, [U_2,U_1] \in \H$,  then $U_2 < [U_2, U_1] $.
\end{enumerate}
Observe that $\H_n = \{ U \in \H\ : \ \order(U)=n\}$.

Clearly, one can construct a totally ordered set $\H$ satisfying conditions 1 and 2 by inductively determining $\H_{n}$ (for $n\geq 2$) from condition 2 provided a total order on $\cup_{k=1}^{n-1} \H_k$ has been chosen, and then arbitrarily extending that order relation to $\cup_{k=1}^{n} \H_k$.  However, that construction of $\H$ does not guarantee in general the fulfillment of condition 3. 

If instead of arbitrarily extending the total ordering of $\cup_{k=1}^{n-1} \H_k$ to $\cup_{k=1}^{n} \H_k$, one imposes that $U_1 < U_2$ provided that $\order(U_1) < \order(U_2)$, then condition 3 is automatically satisfied. 

A different generalized Hall  basis $\H$ for the free Lie algebra $L(A,B)$ that will allow us to construct a basis for the Lie algebra $L_{\P}(A,B)$ of classical mechanics can be constructed by imposing a different condition to the total order relation that also guarantees the fulfillment of condition 3 provided that condition 1 and 2 are satisfied: We require that $U_1 < U_2$ provided that either 
\begin{itemize}
\item $\degree(U_1) < \degree(U_2)$, or
\item $\degree(U_1) = \degree(U_2)$ and $\order(U_1) < \order(U_2)$.
\end{itemize}
where $\degree(\cdot)$ refers to the grading by degree of the free Lie algebra $L(A,B)$ given in Definition~\ref{def:gradingbydegree}.  Obviously, there is much freedom in extending the total ordering of $\cup_{k=1}^{n-1} \H_k$ to $\cup_{k=1}^{n} \H_k$ for each $n \geq 2$ in the inductive construction of the generalized Hall basis $\H$. In what follows, we assume that a particular choice has been made to give a total ordering to the elements with a common degree and order. For instance, such a total order can be defined as follows:  if $U=[U_1,U_2]$, $U '=[U'_1,U'_2]$, and $\degree(U)=\degree(U')$ and $\order(U)=\order(U')$, then $U > U'$ provided that either (i) $U_2 > U'_2$,  or (ii)  $U_2 = U'_2$ and $U_1 > U'_1$.  The elements of order up to six of the corresponding generalized Hall basis are displayed (sorted according to the total order in $\H$) in Table~\ref{tab:1}.

\begin{table}
\begin{center}
\begin{tabular}{|crr|}
\hline
$U$ & $\order(U)$ &$\degree(U)$\cr
\hline
$B $                      &      1     &           0  \cr
$[B,[B,A]]  $            &      3     &           0  \cr
$\gray{[B,[B,[B,A]]]} $        &      4     &   0  \cr
$ [[B,[B,A]],[B,A]]   $   &      5    &          0  \cr
$\gray{ [B,[B,[B,[B,A]]]]} $     &      5   &           0 \cr
$\gray{[[B,[B,[B,A]]],[B,A]]} $ &      6   &           0\cr 
$\gray{[B,[B,[B,[B,[B,A]]]]]} $ &      6   &           0\cr 
$[B,A]          $         &      2     &           1  \cr
$[[B,[B,A]],A]     $      &      4     &           1  \cr 
$\gray{[[B,[B,[B,A]]],A]} $&      5     &      1  \cr 
$   [[[B,[B,A]],[B,A]],A ] $&     6   &         1  \cr
$ \gray{[[B,[B,[B,[B,A]]]],A]} $&     6   &         1  \cr
$ [[B,A],[[B,[B,A]],A]]$  & 6     &  1   \cr   
$A   $                    &      1    &           2  \cr
$[[B,A],A]       $        &      3     &           2  \cr   
$ [[[B,[B,A]],A],A]  $    &      5    &          2   \cr
$ [[B,A],[[B,A],A]]  $    &     5    &          2   \cr
$\gray{ [[[B,[B,[B,A]]],A],A]} $&     6    &          2   \cr
$ [A,[[B,A],A]]  $        &      4   &          3  \cr
 $ [A,[[[B,[B,A]],A],A]]$ & 6 	&  3 \cr
$ [A,[[B,A],[[B,A],A]]] $ & 6	& 3  \cr
$  [A,[A,[[B,A],A]]]  $    &     5   &          4  \cr
$  [A,[A,[A,[[B,A],A]]]]$ & 6	& 5  \cr
\hline
\end{tabular}
\caption{Elements of order up to 6 of a generalized Hall basis $\H$ for $L(A,B)$. The elements in the Lie ideal $\mathcal{I}$  generated by $[L^0(A,B),L^0(A,B)]$ are depicted in gray. The elements belonging to the basis $\B$ for $L_{\P}(A,B) = L(A,B)/\mathcal{I}$ are depicted in black. 
The elements are listed in the total order defined in \S\ref{ss:basis}.
\label{tab:1}}
\end{center}
\end{table}

A basis  $\B$ for the Lie algebra of classical mechanics $L_{\P}(A,B)$ can be  constructed as a subset of such a generalized Hall  basis $\H = \cup_{n\geq 1} \H_{n}$. 
Indeed, $\B = \cup_{n\geq 1} \B_{n}$ where  $\B_1 = \{A,B\}$ and for $n \geq 2$,
\begin{equation*}
\B_n = \{U \in \H_n \ : \  U = [U_1, U_2], \ U_1,U_2 \in \B, \ \degree(U_2) > 0\}.
\end{equation*}

\begin{proof}{of Theorem~\ref{th:1}}
Let 
\begin{equation*}
\B^0 := \{X \in \B\ : \ \degree(X)=0\}
\end{equation*}
and
\begin{equation*}
\overline{\B} := \{X \in \B\ : \ \degree(X)>0\}.
\end{equation*}
Then 
the elements of $\B^0$ of positive order are of the form $X = [X_2 X_1 A]$, where $X_1,X_2 \in\B^0$,
and  the elements of $\overline{\B}$ of positive order either 
\begin{itemize}
\item[(i)] belong to  $\Y = [\X , A]$, or
\item[ (ii)] are of the form $[Y_2, Y_1]$ where $Y_1,Y_2 \in \Y$ and  $Y_2 < Y_1$, or
\item[ (iii)]  are of the form $[Y,  A]$ with $Y \in \Y$, or 
\item [(iv)] of the form $[U_3 U_2 U_1]$ where $U_1,U_2,U_3,[U_2, U_1] \in \overline{\B}$, and $[U_2,  U_1] > U_3 \geq U_2$. 
\end{itemize}
This shows that  $\B^0$ coincides with the set $\X$ introduced in Theorem~\ref{th:1}, and that 
$\overline{\B}$ is actually a generalized Hall basis of the free Lie algebra $L(A, [\X,A])$ over the set $\{A\} \cup [\X,A]$. \end{proof}

\subsection{Realization of the Lie algebra of classical mechanics}

\label{ss:real}

Given a particular mechanical system with potential energy $V(q)$ and kinetic energy $T(p,q) = 1/2 \, M(q)(p,p)$, the Lie algebra generated under the Poisson bracket by $V$ and $T$ will typically have linear dependencies than are not present in $L_\P(A,B)$, that is, $\ker \Phi_{T,V} \neq 0$. However, we will show
in this section that there are no linear dependencies that are shared by all possible mechanical systems with arbitrary degrees of freedom, that is,  the intersection of the kernels of all the homomorphisms
$\Phi_{T,V}$ corresponding to all smooth mechanical systems reduces to
$0$.  
Specifically, we will show in Theorem~\ref{th:2} below that $L_\P(A,B)$ can be realized as the projective limit of the Lie algebras of a sequence of mechanical systems with polynomial Hamiltonian functions. 

This is achieved by constructing (along the lines of the standard proof of independence of elementary differentials) for each $n$ specific kinetic and potential energy functions such that the Lie algebra they generate has no additional identities up to order $n$. 

Let $\mathcal{B} = \{U_1,U_2, U_3, \ldots\}$ be a totally ordered basis for $L_\P(A,B)$, as constructed in Subsection~\ref{ss:basis} (in particular, satisfying that $U_i < U_j$ if $\mathrm{degree}(U_i) < \mathrm{degree}(U_j)$), 
 such that, $U_1=A$, $U_2=B$, and for each $n\geq 1$,  the elements of $\{U_i\ :  d_{n-1} < i \leq d_n\}$ are of order $n$.
 
Let for each $n\geq 1$ denote as  $L_\P(A,B; n)$ the subspace of  $L_\P(A,B)$ (of dimension $d_n$) spanned by  $\{U_i\ :  i\leq d_n\}$. For each $i \geq 3$, we will define a monomial $W_i(p,q)$ in the variables $p_0,p_1,\ldots,p_{d_n}, q_1,\ldots,q_{d_n}$ as follows. 
\begin{itemize}
%\item $W_1 = p_0 p_1$,
%\item $W_2 = q_2$,
 \item $W_i(p,q) = q_i p_{j_1} p_{j_2}$ if $U_i = [U_{j_2}  U_{j_1} A]$, $\mathrm{degree}(U_{j_1})=\mathrm{degree}(U_{j_2})=0$,
\item $W_i(p,q) = p_i p_j$ if $U_i = [U_{j} , A]$, $\mathrm{degree}(U_{j}) = 0$,
\item $W_i(p,q) = p_i p_{j_1} q_{j_2} \cdots q_{j_m}$ if $U_i = [U_{j_m} \cdots U_{j_1} A]$, $m\geq 2$,  $\mathrm{degree}(U_{j_1})=0$, $\mathrm{degree}(U_{j_2})>0$,
\item $W_i(p,q) = p_0 p_i q_{j_1} \cdots q_{j_m}$ if $U_i = [U_{j_m} \cdots U_{j_1} A]$, $m\geq 1$, $\mathrm{degree}(U_{j_1})>0$.
\end{itemize}

\begin{theorem}
\label{th:2}
For each $n\geq 1$, consider the mechanical system (with $d_n+1$ degrees of freedom) with $V(q) = q_2$ and 
\begin{equation*}
T(p,q) = p_0 p_1 + \sum_{i=2}^{d_n}  W_i(p,q),
\end{equation*}
where $p= (p_0,p_1,\ldots,p_{d_n})$ and $q= (q_0,q_1,\ldots,q_{d_n})$.
Then,
\begin{equation*}
 \ker \Phi_{T,V} \bigcap  L_\P(A,B;n) = 0.
\end{equation*}
\end{theorem}

Theorem~\ref{th:2} is a direct consequence of the following result.

\begin{proposition}
\label{prop:1}
Under the assumptions of Theorem~\ref{th:2}, let  us consider $p^0 \in \R^{d_n+1}$ and $q^0 \in \R^{d_n}$ given by $p^0 = (1,0,\ldots,0)$ and $q^0 = (0,\ldots,0)$.  Given $i , j \leq d_n$, 
\begin{equation*}
\left( \frac{\partial }{\partial p_i} +\frac{\partial }{\partial q_i} \right) \Phi_{T,V}(U_j)
\end{equation*}
evaluated at $(p,q)=(p^0,q^0)$ is nonzero if and only if $i=j$.
\end{proposition}

We omit the proof of Proposition~\ref{prop:1}, which is rather technical, and very similar to the proof of Lemma~2.3 in~\cite{grossman}.

\section{The case of classical mechanical systems in Euclidean space}
\label{s:euclidean}

We now consider  classical mechanical systems with kinetic energy $T(p)=\frac{1}{2} \, p \cdot p$.  From now on, given a smooth potential $V:\R^d\to\R$, we will denote by $\Phi_{V}$ the homomorphism $\Phi_{T,V}$ given in Definition~\ref{def:Phi} for $T(p) = \frac{1}{2} \, p \cdot p$ ($p \in \R^d$).

As in \S\ref{ss:real}, given a particular Euclidean mechanical system with potential $V(q)$, the Lie algebra $\Phi_{V}(L_\P(A,B))$ generated under the Poisson bracket by $V$ and $T$ will typically have linear dependencies than are not present in $L_\P(A,B)$, that is, $\ker(\Phi_{V}) \neq 0$. 

We conjecture that there are no relations other than those inherited from $L_\P(A,B)$ (antisymmetry, Jacobi identity, and that the vanishing of Lie bracket of two elements of degree 0) that are shared by all possible smooth Euclidean mechanical systems with arbitrary dimensions. More precisely,
\begin{conjecture}
\label{conj:ker0}
Consider all smooth potentials $V \in \Cinf$ with arbitrary $d$.
Then it holds that
\begin{eqnarray*}
\bigcap_{d\geq 1} \bigcap_{V \in \Cinf} \ker(\Phi_{V}) = 0.
\end{eqnarray*}
\end{conjecture}
We will next give some hints that we believe could be helpful for a possible proof of the conjecture.

Let $\T$ be the set of free trees with two types of vertices, thick ($\o$) and thin ($\cdot$), such that vertices with more than one edge are thick. Then 
for each $W \in L_\P(A,B)$, $\Phi_V(W)$ can be written as a linear combination with integer coefficients of certain functions (referred to as elementary Hamiltonians) associated to each tree in $\T$. (These trees
are an alternative representation of the ``free RKN trees'' of \cite{Calvo,murua}.)

It is useful to think of each edge connecting a thick vertex to a thin vertex as a {\em free-end} edge; such a  free-end edge is ``ready" to be grafted to (fill its free end with) a thick vertex of another tree. 

 Given $u \in \T$, $\order(u)$ is the sum of the number of free-end edges plus twice the number of thick vertices minus one. As for the degree: $\degree(\cdot)=2$, and for $u \in \T\backslash\{\cdot\}$, $\degree(u)$ is the number of free-end edges of $u$.

We next define a binary operation  $\lhd$ on $\span(\T)$. 
 \begin{definition}
 $\lhd: \span(\T) \otimes  \span(\T) \to  \span(\T)$ is defined as follows: 
 $\o \lhd \cdot= \oe$,   $u \lhd \o = 0$,  $\cdot \lhd u = 0$  for all $u \in \T$.
If $u \in \T$ is a tree with $\degree(u)=m$ and  $k\geq 1$ thick vertices, then 
\begin{itemize}
\item[(i)] $u \lhd \cdot$ is the sum of the $k$ trees with $k$ thick vertices and degree $m+1$ obtained by adding a free-end edge to one thick vertex of $u$. 
\item[(ii)] $\o \lhd u$ is the sum of the $m$ trees with $k+1$ thick vertices and degree $m-1$ obtained by grafting one free-end edge of $u$ to a new thick vertex. 
\item[(iii)] Given two trees $u, v \in \T\backslash\{\o,\cdot\}$, let $k$ be the number of thick vertices of $u$ and let $m=\degree(v)$. Then $u \lhd v \in \T$ is the sum of the $k m$ trees obtained by grafting  a free-end edge of $v$ to a thick vertex of $u$.
\end{itemize}
\end{definition}
Note that $$\order(u \lhd v) = \order(u) + \order(v) \quad \mbox{and} \quad \degree(u \lhd v) = \degree(u) + \degree(v)-1.$$ 

%We denote $\T_n = \{u \in \T\ : \ \degree(u)=n\}$ for each $n\geq 1$. 

\begin{proposition}
\label{prop:elHam}
%\begin{enumerate}
The binary operation $[u,v] = u \lhd v  - v \lhd u$ (for $u,v \in \span(\T)$) endows the  
 vector space $\span(\T)$ with a Lie algebra structure of class $\P$. 
 %\item 
 For each $V \in \Cinf$, there exists a homomorphism 
\begin{equation*}
\Psi_V: \span(\T) \rightarrow \Cinfpol
\end{equation*}
of Lie algebras of class $\P$ such that $\Psi_V(\cdot)=T$ and $\Psi_V(\o)=V$.
%\item 
By the universal property of $L_\P(A,B)$, 
 \begin{itemize}
\item[(i)] there exists a unique homomorphism 
\begin{equation*}
\Theta:L_\P(A,B) \rightarrow \span(\T)
\end{equation*}
of Lie algebras of class $\P$ such that $\Theta(A)=\o$ and $\Theta(B)=\cdot$, and
\item[(ii)] $\Phi_V = \Psi_V \circ \Theta.$
\end{itemize}
%\end{enumerate}
\end{proposition}

The Lie algebra homomorphism $\Psi_V$ in the statement of Proposition~\ref{prop:elHam} can be defined as follows.
\begin{definition}
\label{def:elham}
Given $V \in \Cinf$, for each $u \in \T$, the {\em elementary Hamiltonian} $\Psi_V(u)$ can be defined as follows: 
$\Psi_V(\o)=V$,  $\Psi_V(\cdot) = T$, and if $u \in \T$ has $m+1$ vertices, then 
\begin{equation*}
\Psi_V(u) = \sum_{j_1,\ldots,j_m=1}^{d} \prod_{i=1}^{m+1} H(j_1,\ldots,j_m)^{[i]}
\end{equation*}
where each of the factors $H(j_1,\ldots,j_m)^{[i]}$ is associated to a vertex in $u$. More precisely, let us label each vertex of $u$ from $1$ to $m+1$ and each edge of $u$ with a different dummy index $j_1,\ldots,j_m$. If the $i$th vertex is thin and  the $\ell$th edge is connected to  some thick vertex, then $H(j_1,\ldots,j_m)^{[i]} = p_{\ell}$. If the $i$th vertex is thick and has $r$  edges labelled by $\ell_1,\ldots,\ell_r$, then 
$$H(j_1,\ldots,j_m)^{[i]} = V_{\ell_1\ldots\ell_r} := \frac{\partial^r V}{\partial q_{\ell_1} \cdots\partial  q_{\ell_r}}.$$
\end{definition}

 In Table~\ref{tab:2}, the images by $\Phi_V\colon L_\P(A,B)  \rightarrow \Cinfpol$ and $\Theta\colon L_\P(A,B) \rightarrow \span(\T)$ of elements $W$ of a basis of $L_\P(A,B)$ of order up to 6 are displayed. The images by  
$\Psi_V:\span(\T) \rightarrow \Cinfpol$ can be determined from the identity $\Phi_V = \Psi_V \circ \Theta$.

 \renewcommand\arraystretch{1.5}
\begin{table}
\begin{center}
\begin{tabular}{|cccl|}
\hline
$U$ & $\degree(U)$  & $\Theta(U)$ &$\Phi_V(U)$\cr
\hline
$B$                         &  0 &     \o     &          V  \cr \hline
$[B,[B,A]]$              &  0 &   \oo      &            $V_i V_i$  \cr \hline
$ [[B,[B,A]],[B,A]]$  &   0 & $2\, \ooo $   &          $2\,\ V_i V_{ij} V_j  $ \cr \hline
$[B,A]  $                 &   1 &  $-\oe$     &            $-\,V_i\, p_i$   \cr \hline
$[[B,[B,A]],A] $          &  1 &    $2\, \ooe$     &           $2 V_i\, V_{ij} \, p_j$ \cr \hline 
  $ [[[B,[B,A]],[B,A]],A ] $&  1 &   $4\, \oooe$   &  $4  V_i\, V_{ij}\, V_{jk}\, p_k$    \cr
                                  &    &      $+ 2\, \oope$                       &  + $2 p_i\,  V_{ijk}\,  V_j \, V_k $ \cr \hline
$ [[B,A],[[B,[B,A]],A]]$  & 1 &  2\oope    &   $2 p_i\,  V_{ijk} \,  V_j \, V_k $ \cr \hline   
$A $                      &     2 &  $\cdot$    &           $T$  \cr \hline
$[[B,A],A] $              &    2 &   \eoe  &   $   p_i\, V_{ij}  \, p_j$  \cr \hline   
$ [[[B,[B,A]],A],A] $     &  2 &  $2\,\eooe$     &    $2 p_k\, V_{ki}\, V_{ij} \, p_j$         \cr
                                &    &      $-2\,\eope$   & $-2 p_j\, V_i \, V_{ijk} \, \,  p_k$  \cr \hline
$ [[B,A],[[B,A],A]] $     &   2 &     $2\,\eooe$     &        $2 p_k\, V_{ki}\, V_{ij} \, p_j$    \cr
                                &    &    $+\,\eope$      & $+ p_j\, V_i \, V_{ijk} \, \,  p_k$  \cr \hline
$ [A,[[B,A],A]] $         &     3 &   $\eepe$     &       $ p_i\, p_j\, V_{ijk} \,  p_k$    \cr \hline
  $[A,[[[B,[B,A]],A],A]]$ &   3 &  $-6\, \eoepe$ &   $-6  p_l\, V_{li}\, V_{ijk} \, p_j\,  p_k$  \cr
                                 &      &  $-2\, \epepe$ & $-2 p_j\, V_i \, V_{ijkl}\,  p_k\, p_l$ \cr \hline
$ [A,[[B,A],[[B,A],A]]] $ &  3 &  $-3\, \eoepe$	&  $-3  p_l\, V_{li}\, V_{ijk} \, p_j\,  p_k$  \cr
                                 &      &    $+ \epope$   &  $+ p_j\, V_i \, V_{ijkl}\,  p_k\, p_l$ \cr \hline
 $ [A,[A,[[B,A],A]]] $     &   4 &    $ \epepe$  &  $ p_l\, p_k\,  p_j\, p_i\,  V_{ijkl}$    \cr \hline
 $ [A,[A,[A,[[B,A],A]]]] $&   5 &  \epoepe	&      $ p_m\, p_l\, p_k\,  p_j\, p_i\, V_{ijklm}$      \cr
\hline
\end{tabular}
\caption{Image by $\Theta$ and $\Phi_V$ of Prop. \ref{prop:elHam} of elements in $L_\P(A,B)$ of order up to 6. Repeated indices are summed from 1 to $d$.
\label{tab:3}}
\end{center}
\end{table}

The following result was proved (using a different notation and terminology) by Calvo in~\cite{Calvo}.
\begin{proposition}
\label{prop:calvo}
  Consider all polynomial potentials $V:\R^d\to \R$ with arbitrary $d$.
Then it holds that
\begin{eqnarray*}
 \bigcap_{V} \ker(\Psi_{V}) = 0.
\end{eqnarray*}
\end{proposition}
\begin{proof}{}
It is sufficient to show that for each $u \in \T$, there exists $V \in \Cinf$,  with $d$ the number of edges of $u$,  such that, given $u' \in \T$, the value of $\Psi_V(u')$ at $p=(1,\ldots,1)$ and $q=(0,\ldots,0)$ is non-zero if and only if $u'=u$. Such a function $V$ can be constructed for each $u \in \T$ as follows: Label each thick vertex of $u$ from $1$ to $m$ ($m=d+1-\degree(u)$), and each edge from $1$ to $d$. Then, $V = W_1+\cdots + W_{d+1}$, where $W_i = q_{\ell_1} \cdots q_{\ell_r}$ provided that the edges adjacent to the $i$th vertex are those labelled by $\ell_1,\ldots,\ell_r$ respectively. 
\end{proof}  

Conjecture~\ref{conj:ker0} is now equivalent to the following 
  \begin{conjecture}
\label{conj:Theta}
$\Theta:L_\P(A,B) \to \span(\T)$ is injective.
\end{conjecture}

To support our conjecture, we compare the dimensions of the homogeneous subspaces $\T^n$ of elements of $\T$ of order $n=1,2,3,\ldots,20$, namely
 \begin{center}
 2, 1, 2, 2, 4, 5, 10, 14, 27, 43, 82, 140, 269, 486, 939, 1765, 3446, 
6652
\end{center}
with the dimensions of $L_\P^n(A,B)$ displayed in Table~\ref{tab:2}. They coincide up to $n=8$, and obey
$\dim L_\P^n(A,B)<\dim \T^n$ for $n>8$. The  dimensions of the subspaces of homogeneous order and degree of $L_P(A,B)$ and $\T$ (displayed in Table~\ref{tab:dimLACM} and Table~\ref{tab:dimT}) are also compatible
with the conjecture. 

\begin{table}
\begin{center}
\tabcolsep=0pt\def\arraystretch{1.1}
\begin{tabular}{*{19}{wr{9mm}}}
\small
%{cccccccccccccccccc}
 1 & 0 & 1 & 0 & 0 & 0 & 0 & 0 & 0 & 0
   & 0 & 0 & 0 & 0 & 0 & 0 & 0 & 0 \\
 0 & 1 & 0 & 0 & 0 & 0 & 0 & 0 & 0 & 0
   & 0 & 0 & 0 & 0 & 0 & 0 & 0 & 0 \\
 1 & 0 & 1 & 0 & 0 & 0 & 0 & 0 & 0 & 0
   & 0 & 0 & 0 & 0 & 0 & 0 & 0 & 0 \\
 0 & 1 & 0 & 1 & 0 & 0 & 0 & 0 & 0 & 0
   & 0 & 0 & 0 & 0 & 0 & 0 & 0 & 0 \\
 1 & 0 & 2 & 0 & 1 & 0 & 0 & 0 & 0 & 0
   & 0 & 0 & 0 & 0 & 0 & 0 & 0 & 0 \\
 0 & 2 & 0 & 2 & 0 & 1 & 0 & 0 & 0 & 0
   & 0 & 0 & 0 & 0 & 0 & 0 & 0 & 0 \\
 2 & 0 & 4 & 0 & 3 & 0 & 1 & 0 & 0 & 0
   & 0 & 0 & 0 & 0 & 0 & 0 & 0 & 0 \\
 0 & 4 & 0 & 6 & 0 & 3 & 0 & 1 & 0 & 0
   & 0 & 0 & 0 & 0 & 0 & 0 & 0 & 0 \\
 3 & 0 & 9 & 0 & 8 & 0 & 4 & 0 & 1 & 0
   & 0 & 0 & 0 & 0 & 0 & 0 & 0 & 0 \\
 0 & 9 & 0 & 14 & 0 & 11 & 0 & 4 & 0 &
   1 & 0 & 0 & 0 & 0 & 0 & 0 & 0 & 0
   \\
 6 & 0 & 20 & 0 & 23 & 0 & 14 & 0 & 5
   & 0 & 1 & 0 & 0 & 0 & 0 & 0 & 0 & 0
   \\
 0 & 18 & 0 & 37 & 0 & 32 & 0 & 17 & 0
   & 5 & 0 & 1 & 0 & 0 & 0 & 0 & 0 & 0
   \\
 11 & 0 & 46 & 0 & 62 & 0 & 46 & 0 &
   21 & 0 & 6 & 0 & 1 & 0 & 0 & 0 & 0
   & 0 \\
 0 & 41 & 0 & 90 & 0 & 97 & 0 & 60 & 0
   & 25 & 0 & 6 & 0 & 1 & 0 & 0 & 0 &
   0 \\
 23 & 0 & 106 & 0 & 165 & 0 & 144 & 0
   & 80 & 0 & 29 & 0 & 7 & 0 & 1 & 0 &
   0 & 0 \\
 0 & 88 & 0 & 228 & 0 & 274 & 0 & 206
   & 0 & 100 & 0 & 34 & 0 & 7 & 0 & 1
   & 0 & 0 \\
 46 & 0 & 248 & 0 & 438 & 0 & 438 & 0
   & 285 & 0 & 127 & 0 & 39 & 0 & 8 &
   0 & 1 & 0 \\
 0 & 198 & 0 & 562 & 0 & 777 & 0 & 658
   & 0 & 384 & 0 & 154 & 0 & 44 & 0 &
   8 & 0 & 1 \\
\end{tabular}
\caption{
\label{tab:dimLACM}
Dimensions of homogeneous subspaces of $L_\P(A,B)$ of order $n$ (rows $n=1,\ldots,18$) and degree $m-1$ (columns $m=1,\ldots,18$).
}
\end{center}
\end{table}

\begin{table}
\begin{center}
\tabcolsep=0pt\def\arraystretch{1.1}
\begin{tabular}{*{19}{wr{9mm}}}
\small
 1 & 0 & 1 & 0 & 0 & 0 & 0 & 0 & 0 & 0
   & 0 & 0 & 0 & 0 & 0 & 0 & 0 & 0 \\
 0 & 1 & 0 & 0 & 0 & 0 & 0 & 0 & 0 & 0
   & 0 & 0 & 0 & 0 & 0 & 0 & 0 & 0 \\
 1 & 0 & 1 & 0 & 0 & 0 & 0 & 0 & 0 & 0
   & 0 & 0 & 0 & 0 & 0 & 0 & 0 & 0 \\
 0 & 1 & 0 & 1 & 0 & 0 & 0 & 0 & 0 & 0
   & 0 & 0 & 0 & 0 & 0 & 0 & 0 & 0 \\
 1 & 0 & 2 & 0 & 1 & 0 & 0 & 0 & 0 & 0
   & 0 & 0 & 0 & 0 & 0 & 0 & 0 & 0 \\
 0 & 2 & 0 & 2 & 0 & 1 & 0 & 0 & 0 & 0
   & 0 & 0 & 0 & 0 & 0 & 0 & 0 & 0 \\
 2 & 0 & 4 & 0 & 3 & 0 & 1 & 0 & 0 & 0
   & 0 & 0 & 0 & 0 & 0 & 0 & 0 & 0 \\
 0 & 4 & 0 & 6 & 0 & 3 & 0 & 1 & 0 & 0
   & 0 & 0 & 0 & 0 & 0 & 0 & 0 & 0 \\
 3 & 0 & 10 & 0 & 9 & 0 & 4 & 0 & 1 &
   0 & 0 & 0 & 0 & 0 & 0 & 0 & 0 & 0
   \\
 0 & 9 & 0 & 17 & 0 & 12 & 0 & 4 & 0 &
   1 & 0 & 0 & 0 & 0 & 0 & 0 & 0 & 0
   \\
 6 & 0 & 24 & 0 & 30 & 0 & 16 & 0 & 5
   & 0 & 1 & 0 & 0 & 0 & 0 & 0 & 0 & 0
   \\
 0 & 20 & 0 & 50 & 0 & 44 & 0 & 20 & 0
   & 5 & 0 & 1 & 0 & 0 & 0 & 0 & 0 & 0
   \\
 11 & 0 & 63 & 0 & 96 & 0 & 67 & 0 &
   25 & 0 & 6 & 0 & 1 & 0 & 0 & 0 & 0
   & 0 \\
 0 & 48 & 0 & 146 & 0 & 164 & 0 & 91 &
   0 & 30 & 0 & 6 & 0 & 1 & 0 & 0 & 0
   & 0 \\
 23 & 0 & 164 & 0 & 315 & 0 & 267 & 0
   & 126 & 0 & 36 & 0 & 7 & 0 & 1 & 0
   & 0 & 0 \\
 0 & 115 & 0 & 437 & 0 & 592 & 0 & 408
   & 0 & 163 & 0 & 42 & 0 & 7 & 0 & 1
   & 0 & 0 \\
 47 & 0 & 444 & 0 & 1022 & 0 & 1059 &
   0 & 603 & 0 & 213 & 0 & 49 & 0 & 8
   & 0 & 1 & 0 \\
 0 & 286 & 0 & 1300 & 0 & 2126 & 0 &
   1754 & 0 & 856 & 0 & 265 & 0 & 56 &
   0 & 8 & 0 & 1 \\
\end{tabular}
\caption{Dimensions of homogeneous subspaces of $\T$ of order $n$ (rows $n=1,\ldots,18$) and degree $m-1$ (columns $m=1,\ldots,18$).
\label{tab:dimT}
}
\end{center}
\end{table}

To further support our conjecture, we next show that the restriction of $\Theta$ to $[\X,A]$ (and thus also the restriction to $\X$) is injective: Clearly, $\T^0$ can be identified with the set of (uncoloured) free trees. Hence, for each $X \in \X$, $\Theta(X)$ can be seen as a linear combination of free trees. On the other hand, $\T^1$ can be identified with the set of uncoloured rooted trees, each tree $u \in T^1$ with $m$ vertices being identified with a rooted tree with $m-1$ vertices (the unique thin vertex in $u$ indicating the location of the root). With that identification, the restriction of the binary operation $\lhd$ to $\T^1$ corresponds to the grafting operation on rooted trees. Hence, according to~\cite{chapoton}, $\span(\T^1)$ is the free pre-Lie algebra on one generator with respect to the operation $\lhd$. Now, it can be shown that, given $X_1,X_2 \in \X$, 
 $$\Theta([X_1 * X_2, A]) = \Theta([X_1,A]) \rhd \Theta([X_2,A]) + \Theta([X_2,A]) \rhd \Theta([X_1,A]).$$
The injectivity of the restriction of $\Theta$ to $[\X,A]$ is then a consequence of the main result in~\cite{BergeronLoday}.

\section{The Lie algebra of the time-dependent Schr\"odinger equation}
\label{s:schrodinger}

Application of the operator splitting method to the time-dependent linear Schr\"odinger equation
$$i \frac{\partial}{\partial t} u(q,t) = \nabla^2 u(q,t)  +  V(q) u(q,t)$$
in $\R^d$ leads to the Lie algebra of endomorphisms of $\Cinf$ generated under the commutator bracket, by the Laplace operator $\nabla^2 = \sum_{j=1}^d \frac{\partial^2 }{\partial q_j^2}$  and the multiplicative operator $V(q)$. We will show that this Lie algebra is isomorphic to the Lie algebra of classical mechanics in Euclidean space considered in previous subsection. 

Consider the unique Lie algebra morphism from $L(A,B)$ to $\End$ that sends $A$ to $\frac12\, \nabla^2$ and $B$ to $V$. We will say that an endomorphism of $\Cinf$ is of degree 0 if it is a multiplicative operator, and  of degree $n\geq 1$ if it is a differential operator of order $n$. For $n\geq 0$, we denote as $\End_n$ the vector subspace of endomorphisms of degree $n$. One can check that the homomorphism $L(A,B) \to \End$ sends elements of degree 0 to endomorphisms of degree 0. In addition, its kernel includes the Lie ideal $\mathcal{I}$ generated by Lie-brackets of elements of $L(A,B)$ of degree 0. This implies that there exists a unique Lie algebra homomorphism 
$$\phi_V: L_\P(A,B) \to \End$$ such that $\phi_V(A) = \frac12\, \nabla^2$ and $\phi_V(B)= V$. However,  $\phi_V$ is not compatible with the grading by degree. Indeed, provided that $U \in L_\P(A,B)$ has $\degree(U)\geq 1$,  $\phi_V(U)$ is in general a linear combination of endomorphism of degree $\leq \degree(U)$. 

We now consider the linear map $\hat \phi_V: L_\P(A,B) \to \End$  defined as follows: Given $U \in L_\P(A,B)$ with $\degree(U)=n$, we set $\hat \phi_V(U)$ as the projection to $\End_n$ of $\phi_V$. Clearly, $\hat \phi_V$ is compatible with the grading by degree. Actually, it is a homomorphism of Lie algebras of class $\P$, as it happens to be essentially the same as the homomorphism $\Phi_V:L_\P(A,B) \to \Cinfpol$ considered in previous subsection. More precisely, consider the natural injection $\nu: \Cinfpol \to \End$ that replaces each occurrence of $p_j$ ($j \in \{1,\ldots,d\}$) in $H \in \Cinfpol$ with $\partial_j$;  then 
$$\hat \phi_V = \nu \circ \Phi_V.$$  
This in turn shows that, if Conjecture~\ref{conj:Theta} holds true, then
\begin{eqnarray*}
\bigcap_{d\geq 1} \bigcap_{V \in \Cinf} \ker(\phi_{V}) = 0.
\end{eqnarray*}

\def\arraystretch{1}
\begin{table}
\small
\begin{center}
\begin{tabular}{|rrrrr|}
\hline
$n$ & $\dim L^n(A,B)$&$\dim L^n_\P(A,B)$&$x_n$ & $y_n$ \cr
\hline
1&      2     &           2 & 1 & \cr
2&      1     &           1 &  & 1\cr
3&      2     &           2 & 1 & \cr
4&      3     &           2 & & 2 \cr   
5&      6     &           4 & 1 & \cr
6&      9     &           5 & &  4 \cr
7&      18    &          10 & 2 & \cr
8&      30    &          14 &  & 9 \cr
9&      56    &          25 & 3 & \cr
10&     99    &          39 &  & 18\cr
11&     186   &          69 & 6 & \cr
12&     335   &         110 & & 41\cr
13 & 630	&	193 & 11 & \cr
14 & 1161 	&  320 &  & 88\cr
15 & 2182	& 555 & 23  & \cr
16 & 4080	& 938 &  & 198\cr
17 & 7710	& 1630 & 46  & \cr
18 & 14532	& 2786 & & 441 \cr
19 & 27594	& 4852 & 98 &  \cr
20 & 52377	& 8370 &  & \cr
\hline
\end{tabular}
\caption{Dimensions of Lie algebras graded by order.
Column 2: Of the free Lie algebra with two generators.
Column 3: Of the Lie algebra of classical mechanics, $L_\P(A,B)$.
Column 4: Dimensions $x_n$ of the subspace of elements of degree 0 in  $L^n_\P(A,B)$.
Column 5: Dimensions $y_n$ of the subspace of elements of degree 1 in  $L^n_\P(A,B)$.
\label{tab:2}}

\end{center}
\end{table}

\section{The entropy of the Lie algebra of classical mechanics}
\label{s:entropy}

Let $z(t)$ be the generating function for the enumeration of $\X$ by order. Let $a(t)$ be the generating function for the enumeration of binary trees. It obeys
\begin{equation}
\label{eq:A}
 a(t) = t + \frac{1}{2}(a(t)^2 + a(t^2)) = t + t^2 + t^3 + 2 t^4 + 3 t^5 + 6 t^6 + \dots=\sum_{n=1}^\infty a_n t^n
 \end{equation}
Because the order of an element of $\X$ with $n$ $B$s is $2n-1$, we have
\begin{equation}
\label{eq:zdim}
z(t) = t^{-1} a(t^2).
\end{equation}
For any set $\W$ with generating function 
$w(t)=\sum_{n>0}|\{W\in\W:\order(W)=n\}|t^n$, the dimensions
of the homogeneous components of the
graded Lie algebra $L(\W)=\bigoplus_{n>0}L^n(\W)$ are given by
\cite{ka-ki,mu-ow}
\begin{equation}
\label{eq:cns}
\dim L^n(\W) = \sum_{d|n}\frac{1}{d}\mu(d) [t^{n/d}](-\log(1-w(t))).
\end{equation}

The generating function for the grading by order of
$\{A\}\cup\X$ is $t+t z(t)=t+a(t^2)$. An application of (\ref{eq:cns}) together with the dimensions of the abelian part from (\ref{eq:zdim}) gives the dimensions of $L_\P(A, B)$ as
listed in Table \ref{tab:2}. 

The asymptotic growth of $a_n$ was obtained by Wedderburn \cite{wedderburn} using a method that we review briefly.
He works in terms of $g(t) := 1-a(t)$, which from (\ref{eq:A}) obeys the functional equation
\begin{equation}
\label{eq:g}
 g(t^2) = g(t)^2 + 2t.
 \end{equation}
First he shows that the singularity $r$ of $g(t)$ of smallest modulus is unique, is a branch point of order 2,  and obeys $g(r)=0$. Since $g(r)=0$ we get
$$ g(r^2) = 2r,\quad g(r^4) = 6 r^2,\quad g(r^8) = 38 r^4,\dots,g(r^{2^{k+1}}) = c_k r^{2^{k}}$$
where
$$ c_0 = 2,\quad c_{k+1} = c_k^2+2,\quad k=1,2,3,\dots.$$
Now $\lim_{k\to\infty}g(r^{2^k}) = g(0) = 1$ so $r = \lim_{k\to\infty} c_k^{-2^{-k}}$.
The singularity at $r$ determines the growth rate of the $a_n$, and a little more work \cite{otter} determines the leading constant in 
$$ z_{2n-1} = c_n \sim \eta n^{-3/2} r^{-n},\quad\eta\approx 0.318777,\quad r^{-1}\approx 2.4832535361726368586.$$

Therefore, the number of modified potentials of odd order $n$ is 
$$\dim(\span\X)^n = \mathcal{O}((r^{-1/2})^n)  = {\mathcal O}(1.2416267680863184293^n).$$

From Eq. (\ref{eq:cns}), the asymptotic growth of $\dim L^n(\W)$
is determined by the singularities
of $-\log(1-w(t))$. These correspond to zeros and singularities of
$1-w(t)$. In particular, if $1-w(t)$ has a simple zero at $t=\alpha$
and no other zero with $|t|\le\alpha$, then 
\begin{equation}
\label{eq:cns2}
 \dim L^n(\W) \sim \frac{1}{n}\left(\frac{1}{\alpha}\right)^n
\end{equation}
and the Lie algebra has entropy $1/\alpha$.

In the case of $L(A,[\X,A])$, we have $w(t) = t + a(t^2)$ and the growth rate is determined by the smallest root $\alpha$ of $t + a(t^2) = 1$. This gives $a(\alpha^2) = 1-\alpha$ or $g(\alpha^2) = \alpha$. Applying (\ref{eq:g}) as before leads to
$$ g(\alpha^2) = \alpha,\quad g(\alpha^4) = 3\alpha^2,\quad g(\alpha^8) = 11\alpha^4,\dots,
g(\alpha^{2^{k+1}}) = e_k r^{2^k}$$
where 
$$ e_0 = 1,\quad e_{k+1} = e_k^2 + 2,\ k=1,2,3,\dots.$$
Therefore $\alpha = \lim_{k\to\infty} d_e^{-2^{-k}}.$

Since the dimensions of the free part dominate the abelian part, we have that the entropy of the Lie algebra of classical mechanics is
$$1/\alpha = 1.8249111600523655937\ldots.$$

\bigskip
\noindent{\bf Acknowledgements.} RM was supported in part by the Royal Society Te Ap\={a}rangi.  AM has been  supported in part by project MTM2016-77660-P(AEI/FEDER, UE) from Ministerio de Econom'a, Industria y Competitividad and  by the Basque Government (Consolidated Research Group IT649-13).

\end{document}